%% file: NonCMC-SmallTT-3-arXiv.tex
\documentclass[a4paper,leqno]{amsart}
\usepackage{enumerate,enumitem}
\usepackage{lscape,dsfont}
\usepackage[usenames,dvipsnames]{xcolor}
\usepackage{mathabx,amsmath,amssymb,stmaryrd,mathtools} 
\usepackage{pstricks,pst-node,pst-coil,pst-plot}
\usepackage{tikz}
\usepackage{tikz-cd}

\usetikzlibrary{arrows.meta,decorations.markings,decorations.pathreplacing,patterns}
\usetikzlibrary{matrix}
\usetikzlibrary{graphs}
\usetikzlibrary{backgrounds}

\usepackage[colorlinks=true, bookmarks=true, pdfstartview=FitH, pagebackref=true]{hyperref}
\usepackage[space, compress, sort]{cite}
\usepackage{verbatim}

\theoremstyle{plain}

\newtheorem{thm}{Theorem}[section]
\newtheorem{theorem}[thm]{Theorem}

\newtheorem{corollary}[thm]{Corollary}

\newtheorem{lemma}[thm]{Lemma}

\newtheorem{proposition}[thm]{Proposition}

\newtheorem{claim}{Claim}

\theoremstyle{remark}

\newtheorem{remark}[thm]{Remark}

\theoremstyle{definition}

\newcounter{mnotecount}[section]

\renewcommand{\phi}{\varphi}
\renewcommand{\epsilon}{\varepsilon}

\newcommand{\bR}{\mathbb{R}}
\newcommand{\bL}{\mathbb{L}}

\newcommand{\cY}{\mathcal{Y}}

\newcommand{\scal}{\mathrm{Scal}}

\newcommand{\ghat}{\widehat{g}}
\newcommand{\Khat}{\widehat{K}}
\newcommand{\hscal}{\widehat{\scal}}

\DeclareMathOperator{\tr}{tr}
\DeclareMathOperator{\divg}{div}

\DeclareMathOperator{\vol}{Vol}

\newcommand{\gtil}{\widetilde{g}}
\newcommand{\scaltil}{\widetilde{\scal}}
\newcommand{\phitil}{\widetilde{\phi}}
\newcommand{\psitil}{\widetilde{\psi}}
\newcommand{\Atil}{\widetilde{A}}

\newcommand{\Sring}{\mathring{S}}

\let\<\langle 
\let\>\rangle

\newcommand{\definedas}{\mathrel{\raise.095ex\hbox{\rm :}\mkern-5.2mu=}}

\DeclareMathOperator{\DeltaL}{\Delta_{\bL}}
\DeclareMathOperator*{\essinf}{ess\,inf}

\date{January 15, 2024}



\begin{document}


\author[R. Gicquaud]{Romain Gicquaud}
\address[R. Gicquaud]{Institut Denis Poisson \\ UFR Sciences et Technologie \\
  Facult\'e de Tours \\ Parc de Grandmont\\ 37200 Tours \\ FRANCE}
\email{\href{mailto: R. Gicquaud <Romain.Gicquaud@idpoisson.fr>}{romain.gicquaud@idpoisson.fr}}

\title[Uniqueness for the conformal method]{What Uniqueness for the Holst-Nagy-Tsogtgerel--Maxwell Solutions to the Einstein Conformal Constraint Equations?}

\begin{abstract}
  This paper addresses the issue of uniqueness of solutions in the conformal method for solving the constraint equations in general relativity with arbitrary mean curvature as developed initially by Holst, Nagy, Tsogtegerel and Maxwell. We show that the solution they construct is unique amongst those having volume below a certain threshold.
\end{abstract}

\date{\today}
\maketitle

\tableofcontents

\section{Introduction}
\input{introduction}

\section{Solutions to the conformal constraint equations}\label{secHNT}
\input{hnt}

\section{Uniqueness of the solution under a volume bound}\label{secUniqueness}
\input{uniqueness}

\providecommand{\bysame}{\leavevmode\hbox to3em{\hrulefill}\thinspace}
\providecommand{\MR}{\relax\ifhmode\unskip\space\fi MR }
\providecommand{\MRhref}[2]{%
  \href{http://www.ams.org/mathscinet-getitem?mr=#1}{#2}
}
\providecommand{\href}[2]{#2}


\end{document}

%% file: introduction.tex
The foundational elements in the initial data problem of general relativity are
the Einstein constraint equations. These equations ensure that the initial data
set $(M, \ghat, \Khat)$, consisting of a Cauchy surface $M$ of dimension $n
    \geq 3$, the spatial metric $\ghat$ and extrinsic curvature $\Khat$, is
physically viable and evolves according to Einstein's field equations.
Specifically, the constraint equations are comprised of the Hamiltonian
constraint and the momentum constraint that we are going to describe.

These equations arise regardless of the presence of matter fields, but in this
article, we will focus exclusively on the vacuum case, where only the
gravitational field is modeled. This choice allows for a more streamlined
presentation of the mathematical techniques involved, without additional terms
from matter sources. However, the methods discussed here can be extended to
more general settings with appropriate modifications.

The Hamiltonian constraint ensures that the energy density vanishes on $M$.
Mathematically, it can be expressed as:
\begin{equation}\label{eqHamiltonian}
    0 = \hscal + (\tr_{\ghat} \Khat)^2 - \left|\Khat\right|^2_{\ghat},
\end{equation}

where \( \hscal \) is the scalar curvature of the spatial metric and \(
\tr_{\ghat} \Khat \) is the trace of the extrinsic curvature i.e.~the mean
curvature.

On the other hand, the momentum constraint ensures that the momentum density
vanishes. It is given by:
\begin{equation}\label{eqMomentum}
    0 = \divg_{\ghat} \Khat - d (\tr_{\ghat} \Khat),
\end{equation}

where \( \divg_{\ghat}\) denotes the divergence operator with respect to the
spatial metric $\ghat$. We refer the reader to~\cite{RingstromCauchyProblem}
for an in depth study of the Cauchy problem in general relativity and
to~\cite{CarlottoReview} for an overview of the techniques developed to study
the constraint equations.\\

The conformal method has emerged as a pivotal tool in the initial data
formulation of Einstein's field equations. It consists in transforming the
underdetermined system of the constraint equations~\eqref{eqHamiltonian}
and~\eqref{eqMomentum} into a set of coupled elliptic partial differential
equations by decomposing the initial data $(\ghat, \Khat)$ into seed data
(i.e.~given data) and unknowns that have to be adjusted to satisfy the
constraints. This decomposition allows for a more tractable approach to finding
solutions that satisfy both the geometric and physical requirements of general
relativity. It has been instrumental in advancing our understanding of
spacetime dynamics, particularly in numerical simulations of phenomena such as
black hole mergers and gravitational wave generation, see e.g.~\cite{Lehner}
for an overview.

We now specify the form of the seed data and the unknowns, together with the
Sobolev regularity we will work with throughout the paper:
\begin{itemize}
    \item A closed manifold $M$ of dimension $n \geq 3$,
    \item A given metric $g \in W^{2, p}(M, S_2 M)$ on $M$ for some $p > n/2$ that has
          positive Yamabe invariant: $\cY_g(M) > 0$, volume $1$, and no conformal Killing
          vector field,
    \item A function $\tau \in W^{1, q}(M, \bR)$, with $q \geq p$ and $q > n$, that will
          play the role of the mean curvature of the physical solution $(\ghat, \Khat)$
          embedded in the spacetime solution to Einstein's equations,
    \item A symmetric 2-tensor $\sigma \in L^{2p}(M, S_2 M)$ that is trace-free and
          divergence-free with respect to the metric $g$.
\end{itemize}
And the unknowns are
\begin{itemize}
    \item A positive function $\phi \in W^{2, p}(M, \bR)$,
    \item A vector field $W \in W^{2, p}(M, TM)$.
\end{itemize}

The reasons for these choices for regularity will be made clear in the sequel.
The condition that $(M, g)$ has volume $1$ can be achieved by rescaling the
metric $g$ by a constant factor. It has the useful consequence that, for any
function $f$, the map $r \mapsto \|f\|_{L^r}$ is increasing.\\

For the sake of completeness, we remind the reader that the Yamabe invariant of
the manifold $(M, g)$ is defined as follows:
\[
    \cY_g(M) = \inf_{\substack{u \in W^{1, 2}(M, \bR),\\ u \not\equiv 0}}
    \frac{\int_M \left[\frac{4(n-1)}{n-2} |du|^2 + \scal~u^2\right] d\mu^g}{\left(\int_M u^N d\mu^g\right)^{2/N}},
\]
where \(N \definedas 2n/(n-2)\). This is a conformal invariant, i.e.~it only
depends on the conformal class $[g]$ of the metric $g$, see~\cite{LeeParker}
for an overview.

From the seed data, we construct the metric $\ghat$ and the second fundamental
form $\Khat$ as follows:
\[
    \ghat = \phi^{N-2} g, \quad \Khat = \frac{\tau}{n} \ghat + \phi^{-2} (\sigma + \bL W),
\]
where $\bL$ is the conformal Killing operator defined in coordinates as
follows:
\[
    \bL W_{ij} = \nabla_i W_j + \nabla_j W_i - \frac{2}{n} \nabla^k W_k g_{ij},
\]
$\nabla$ being the Levi-Civita connection of the metric $g$.

The constraint equations are then reformulated in terms of these new variables
into the following system:
\begin{subequations}\label{system}
    \begin{align}
        -\frac{4(n-1)}{n-2} \Delta \phi + \scal \phi & = - \frac{n-1}{n} \tau^2 \phi^{N-1} + \frac{\left|\sigma + \bL W \right|_g^2}{\phi^{N+1}}, \label{eqLichnerowicz} \\
        \DeltaL W                                    & = \frac{n-1}{n} \phi^N d\tau, \label{eqVector}
    \end{align}
\end{subequations}
where we set
\[
    \DeltaL \definedas - \frac{1}{2} \bL^* \left(\bL \cdot\right).
\]
Equation~\eqref{eqLichnerowicz} is known as the Lichnerowicz equation while
Equation~\eqref{eqVector} will be called the vector equation.
Solving~\eqref{system} is then equivalent to solving~\eqref{eqHamiltonian}
and~\eqref{eqMomentum}. Note that we have taken as a convention that all
differential operators and curvature tensors are defined with respect to the
metric $g$. When the objects are defined with respect to a different metric
$\gtil$, we will explicitly indicate it and denote for example $\Delta_{\gtil}$
the Laplace-Beltrami operator or $\scaltil$ the scalar curvature of the metric
$\gtil$.\\

Restricting ourselves to compact Cauchy surfaces, one of the major achievements
of the conformal method has been the classification of constant mean curvature
(CMC) initial data by J. Isenberg in 1995~\cite{Isenberg}. In this particular
case, which is of great physical importance, we have $d\tau \equiv 0$. So
Equation~\eqref{eqVector} reduces to $\DeltaL W = 0$ which imposes that $\bL W
    \equiv 0$, i.e. $W$ is a conformal Killing vector field. This reduces the
system to a single scalar equation, the Lichnerowicz equation, which is more
tractable. We refer the reader to~\cite{GicquaudLichnerowicz} for a simplified
approach to solving this equation. The near-CMC case was subsequently treated
in different papers, in particular by Allen, Clausen and Isenberg
in~\cite{ACI08} for manifolds with positive Yamabe invariant.\\

For arbitrary mean curvature $\tau$, the conformal constraint equations appear
far more intricate and only partial results are known at the moment. A
significant development was introduced by Holst, Nagy and Tsogtgerel in 2008
in~\cite{HNT1, HNT2}. This method was refined shortly after by Maxwell
in~\cite{MaxwellNonCMC} and results with different regularity assumptions were
obtained by Nguyen in~\cite{Nguyen}, see also~\cite{GicquaudHabilitation}. The
strategy to construct solutions to the conformal constraint equations is based
on Schauder's fixed point theorem. The authors construct a compact map $F: X
    \to X$, where $X$ is an appropriate Banach space of functions on $M$, as
follows: given a conformal factor $\phi$, solve the vector
equation~\eqref{eqVector} for the vector $W$ and then use this solution to
obtain a new conformal factor $\psi = F(\phi)$ via the Lichnerowicz
equation~\eqref{eqLichnerowicz}. The difficult part consists in identifying a
closed bounded subset $\Omega$ of $X$ stable by $F$, i.e.~such that $F(\Omega)
    \subset \Omega$.

The success of this method hinges on the assumption of a positive Yamabe
invariant of the seed metric $g$ and the smallness of the TT-tensor $\sigma$ in
some well chosen norm. In this case, following~\cite{GicquaudHabilitation}, the
set $\Omega$ is, roughly speaking, the set of functions $\phi$ whose $L^q$-norm
is less than some small constant, where $q \in (N, \infty]$ depends on the
regularity assumptions of the seed data. We provide below a new derivation of
this result based on techniques introduced in~\cite{Pailleron}. While we do not
claim any novelty here, we intend to obtain an existence result with clearly
identifiable constants.

To state more precisely our results, we introduce the positive constants
$\mu_L$ and $\mu_V$ defined as follows:
\begin{itemize}
    \item $\mu_L > 0$ is a constant such that, if $\phi$ is a solution to the general form of the Lichnerowicz equation
          \[
              - \frac{4(n-1)}{n-2} \Delta \phi + \scal~\phi + \frac{n-1}{n} \tau^2 \phi^{N-1} = \frac{A^2}{\phi^{N+1}},
          \]
          for some $A \in L^2(M, \bR)$, then
          \[
              \mu_L \|\phi^N\|_{L^{\frac{N}{2}+1}}^{2\frac{n-1}{n}} \leq \|A\|_{L^2}^2
          \]
          (see Lemma~\ref{lmLichnerowicz}).
    \item $\mu_V$ is the following Sobolev-type constant:
          \begin{equation}\label{eqYamabeVector}
              \mu_V \definedas \inf_{\substack{V \in W^{1, 2}(M, TM),\\ V \not\equiv 0}}
              \frac{\frac{1}{2}\int_M |\bL V|^2 d\mu^g}{\left(\int_M |V|^N d\mu^g\right)^{2/N}} \geq \mu_0.
          \end{equation}
\end{itemize}

We can now state the first result of this paper:
\begin{theorem}\label{thmA}
    Under the regularity assumptions given above, assume further that $\sigma \not\equiv 0$ is such that
    \begin{equation}\label{eqXmax}
        \left(\frac{2}{\mu_L \mu_V} \frac{n-1}{n} \|d\tau\|_{L^q}^2\right)^{n-1} \|\sigma\|_{L^2}^2 \leq \frac{\mu_L}{n}.
    \end{equation}
    Then there exists a solution $(\phi, W)$ to the system~\eqref{eqLichnerowicz}--\eqref{eqVector} with $\phi \in W^{2, p}(M, \bR)$ and $W \in W^{2, p}(M, TM)$.
\end{theorem}

A notable limitation of Schauder's fixed point theorem is its inability to
guarantee the uniqueness of the solution, even in the stable set $\Omega$. The
paper~\cite{GicquaudNgo} aimed at providing a different point of view on this
method. It leads to similar results but the construction is based on the
implicit function theorem so it provides some sort of local uniqueness.
However, the question of the global uniqueness of the solution pertains. In
particular, the article~\cite{Nguyen2} showed that it can fail. Here we address
this question by introducing a bound on the physical volume of solutions. We
argue that, by restricting our attention to solutions whose associated volume
\[
    \vol_{\ghat}(M) = \int_M \phi^N d\mu^g
\]
is below a fixed threshold $V_{\max}$, we can achieve uniqueness for the
solutions of the conformal constraint equations. The precise statement of our
result is the following:

\begin{theorem}\label{thmB}
    Under the regularity assumptions stated above, assume given a constant $V_{\max} > 0$ such that
    \begin{equation}\label{eqDefDelta}
        \frac{2}{\mu_V} \left(\frac{n-1}{n}\right)^2 \|d\tau\|_{L^n}^2 V_{\max}^{2/n} < \mu_L
    \end{equation}

    where $\mu_L$ and $\mu_V$ are as defined above. Let $\theta$ be a given
    positive constant. There exists a constant $s'_{\max} > 0$ such that if $\sigma
        \not\equiv 0$ is a TT-tensor such that
    \[
        \|\sigma\|_{L^{2p}} < s'_{\max} \quad \text{and} \quad \essinf_M |\sigma| \geq \theta \|\sigma\|_{L^{2p}},
    \]
    there exists a unique solution $(\phi, W)$ to the conformal constraint
    equations~\eqref{system} such that
    \[
        \vol_{\ghat}(M) = \int_M \phi^N d\mu^g  \leq V_{\max}.
    \]
\end{theorem}

We believe that the lower bound assumption on $|\sigma|$ is of a technical
nature and could potentially be removed by a more refined analysis. This point
is discussed further in the concluding remarks.

Interestingly, the technique developed here nearly allows us to replace
Schauder's fixed point theorem by Banach's, making the construction of Holst,
Nagy, Tsogtgerel and Maxwell more robust, in particular for applications in
numerical relativity.\\

The outline of this paper is as follows. In Section~\ref{secHNT}, we provide a
proof of Theorem \ref{thmA} based on ideas from~\cite{Pailleron}. In
particular, this allows us to discuss the size of the solutions with respect to
the seed data. And, in Section~\ref{secUniqueness}, we prove the main result of
the paper, namely Theorem~\ref{thmB}.\\

\noindent\textbf{Acknowledgments:} Part of the work on this paper was done while
the author was in the Guangxi Center for Mathematical Research attending the
conference \emph{International Conference on Geometric Analysis of Ricci Curvature}.
He wants to express his deep gratitude towards the organizers and to Mengzhang Fan
for their hospitality and for providing me an environment conducive to productive work.
The author would also like to thank the referee of this paper for his thorough
proofreading and his comments which allowed the presentation to be significantly
improved.

This research was supported by the ANR grant Einstein-PPF, ANR-23-CE40-0010-03.

%% file: hnt.tex
The aim of this section is to prove Theorem~\ref{thmA}. We start by recalling
facts about the Lichnerowicz and the vector equations that we borrow
from~\cite{Pailleron} but state them here in lesser generality as it will not
be needed here. In particular, in this note, we are interested only in the case
of closed manifolds. This allows several simplifications as compared
to~\cite{Pailleron}. The interest of the method developed there lies in its
ability to work within a weak regularity framework. This simplifies the
construction of the invariant closed set, and allows us to postpone questions
regarding the regularity of the solution.

The following proposition regarding existence and uniqueness for the solution
to the Lichnerowicz equation corresponds to~\cite[Proposition 2.3 and
    2.10]{Pailleron}. The proof is rather lengthy so we refer the reader to the
original paper:

\begin{proposition}\label{propLich}
    Given $A \in L^2(M)$, $A \not\equiv 0$, the Lichnerowicz equation
    \begin{equation}\label{eqLichGeneral}
        - \frac{4(n-1)}{n-2} \Delta \phi + \scal~\phi + \frac{n-1}{n} \tau^2 \phi^{N-1} = \frac{A^2}{\phi^{N+1}}
    \end{equation}
    admits a unique weak positive solution $\phi \in W^{1, 2}(M, \bR)$ in the sense that for all $\psi \in W^{1, 2}(M, \bR) \cap L^\infty(M)$, we have
    \begin{equation}\label{eqLichWeak}
        \int_M \left[\frac{4(n-1)}{n-2} \<d\phi, d\psi\> + \scal~\phi \psi + \frac{n-1}{n} \tau^2 \phi^{N-1} \psi \right] d\mu^g
        = \int_M \frac{A^2}{\phi^{N+1}} \psi d\mu^g.
    \end{equation}
    Further, $\phi^N$ belongs to $L^{\frac{N}{2}+1}(M, \bR)$ and the mapping $\mathrm{Lich}: A \mapsto \phi^N$ is continuous and compact when seen as a mapping
    \[
        \mathrm{Lich}: L^2(M, \bR) \to L^a(M, \bR)
    \]
    for all $\displaystyle a \in \left[1, \frac{N}{2} + 1 \right)$.
\end{proposition}
Of particular interest for us later on is the upper bound for $\|\phi^N\|_{L^{\frac{N}{2}+1}}$ in terms of $\|A\|_{L^2}$. We rederive it here.

From the assumption that $\cY_g(M) > 0$, there exists a function $\psi \in
    W^{2, p}(M, \bR)$, $\psi > 0$, such that the metric $\gtil \definedas
    \psi^{N-2} g$ has scalar curvature $\scaltil \in L^p(M, \bR)$ bounded from
below by a positive constant $\epsilon > 0$, see~\cite{LeeParker}. As $p >
    n/2$, the Sobolev space $W^{2, p}(M, \bR)$ embeds continuously into
$L^\infty(M)$. This implies that $\psi^{-1} \in W^{2, p}(M, \bR)$. If we set
$\phitil \definedas \psi^{-1} \phi$, the Lichnerowicz
equation~\eqref{eqLichWeak} can be written as
\begin{equation}\label{eqLichWeak2}
    \int_M \left[\frac{4(n-1)}{n-2} \<d\phitil, d\psitil\>_{\gtil} + \scaltil~\phitil \psitil + \frac{n-1}{n} \tau^2 \phitil^{N-1} \psitil \right] d\mu^{\gtil}
    = \int_M \frac{\Atil^2}{\phitil^{N+1}} \psitil d\mu^{\gtil},
\end{equation}
for any $\psitil \in W^{1, 2}(M, \bR) \cap L^\infty(M)$ and where $\Atil \definedas \psi^{-N} A$ as follows from a fairly straightforward calculation. The idea is now to choose $\psitil = \phitil^{N+1}$ and integrate over $M$. However, the function $\phitil^{N+1}$ does not belong a priori to $W^{1, 2}(M, \bR) \cap L^\infty(M, \bR)$, so some care is needed. Instead, we choose a cutoff value $k > 0$ and set
\[
    \psitil = \left(\phitil^{N+1}\right)_{, k} \definedas \min \left\{\phitil^{N+1}, k \right\}.
\]
As the mapping $y \mapsto (y^{N+1})_{, k}$ is bounded and Lipschitz, we have
$\left(\phitil^{N+1}\right)_{, k} \in W^{1, 2}(M, \bR) \cap L^\infty(M, \bR)$
providing a legitimate test function for~\eqref{eqLichWeak2}:
\begin{align*}
     & \int_M \left[\frac{4(n-1)}{n-2} \left\<d\phitil, d\left(\phitil^{N+1}\right)_{, k}\right\>_{\gtil} + \scal~\phitil \left(\phitil^{N+1}\right)_{, k}\right] d\mu^{\gtil} \\
     & \qquad = \int_M \left[- \frac{n-1}{n} \tau^2 \phitil^{N-1} + \frac{\Atil^2}{\phitil^{N+1}}\right] \left(\phitil^{N+1}\right)_{, k} d\mu^{\gtil}                         \\
     & \qquad \leq \int_M \Atil^2 \frac{\left(\phitil^{N+1}\right)_{, k}}{\phitil^{N+1}}d\mu^{\gtil}                                                                           \\
     & \qquad \leq \int_M \Atil^2 d\mu^{\gtil}.
\end{align*}
Note that the first term in the left hand side of the previous equation can be rewritten as follows:
\begin{align*}
    \int_M \left\<d \left(\phitil^{N+1}\right)_{, k}, d\phitil\right\>_{\gtil} d\mu^{\gtil}
     & = \int_M \mathds{1}_{\{\phitil^{N+1} \leq k\}}\left\<d \phitil^{N+1}, d\phitil\right\>_{\gtil} d\mu^{\gtil}                                          \\
     & = (N+1) \int_M \mathds{1}_{\{\phitil^{N+1} \leq k\}} \phitil^N \left\<d \phitil, d\phitil\right\>_{\gtil} d\mu^{\gtil}                               \\
     & = \frac{N+1}{\left(\frac{N}{2}+1\right)^2}\int_M \mathds{1}_{\{\phitil^{N+1} \leq k\}}\left|d \phitil^{\frac{N}{2}+1}\right|^2_{\gtil} d\mu^{\gtil}.
\end{align*}
where $\mathds{1}_{\{\phitil^{N+1} \leq k\}}$ is the characteristic function for the set of points $x \in M$ such that $\phitil^{N+1}(x) \leq k$. All in all, we obtain the following inequality:
\[
    \int_M \left[\frac{4(n-1)}{n-2} \frac{N+1}{\left(\frac{N}{2}+1\right)^2} \mathds{1}_{\{\phitil^{N+1} \leq k\}}\left|d \phitil^{\frac{N}{2}+1}\right|^2_{\gtil}+ \scaltil~\left(\phitil^{N+1}\right)_{, k} \phitil\right] d\mu^{\gtil}\\
    \leq \int_M \Atil^2 d\mu^{\gtil}.
\]
As $\scaltil \geq 0$, it follows from the monotone convergence theorem that,
letting $k$ tend to infinity, we have
\begin{equation}\label{eqW1bound}
    \int_M \left[\frac{3n-2}{n-1} \left|d \phitil^{\frac{N}{2}+1}\right|^2_{\gtil}+ \scaltil~\phitil^{N+2}\right] d\mu^{\gtil}\\
    \leq \int_M \Atil^2 d\mu^{\gtil}.
\end{equation}
Since $\scaltil \geq \epsilon > 0$, we obtain the following estimate:
\[
    \int_M \left[\frac{3n-2}{n-1} \left|d \phitil^{\frac{N}{2}+1}\right|^2_{\gtil}+ \scaltil~\phitil^{N+2}\right] d\mu^{\gtil}
    \geq \epsilon \int_M \left[\left|d\phitil^{\frac{N}{2}+1}\right|^2_{\gtil} + \left(\phitil^{\frac{N}{2}+1}\right)^2\right] d\mu^{\gtil},
\]
where we have assumed, without loss of generality, that $\epsilon <
    \frac{3n-2}{n-1}$. From the Sobolev embedding theorem (with respect to the
metric $\gtil$), we know that there exists a constant $\nu > 0$ such that, for
any $u \in W^{1, 2}(M, \bR)$,
\begin{equation}\label{eqDefNu}
    \int_M \left[\frac{3n-2}{n-1} \left|d u\right|^2_{\gtil}+ \scaltil~u^2\right] d\mu^{\gtil}
    \geq \nu \left(\int_M u^N d\mu^{\gtil}\right)^{2/N}.
\end{equation}
In particular, choosing $u = \phitil^{\frac{N}{2}+1}$, we get
\[
    \int_M \left[\frac{3n-2}{n-1} \left|d \phitil^{\frac{N}{2}+1}\right|^2_{\gtil}+ \scaltil~\phitil^{N+2}\right] d\mu^{\gtil}
    \geq \nu \left(\int_M \phitil^{N \left(\frac{N}{2}+1\right)} d\mu^{\gtil}\right)^{2/N}
\]
From Estimate~\eqref{eqW1bound}, we conclude that
\[
    \nu \left(\int_M \phitil^{N \left(\frac{N}{2}+1\right)} d\mu^{\gtil}\right)^{2/N} \leq \int_M \Atil^2 d\mu^{\gtil}.
\]
Written in terms of the metric $g$ and the function $\phi$ the previous
inequality reads
\[
    \nu \left(\int_M \psi^{-N^2/2}\phi^{N \left(\frac{N}{2}+1\right)} d\mu^g\right)^{2/N} \leq \int_M \psi^{-N} A^2 d\mu^g.
\]
Finally, estimating $\psi$ from above and below by its minimum and it maximum,
we conlude that
\[
    \nu \left(\frac{\min_M \psi}{\max_M \psi}\right)^N \left(\int_M \phi^{N \left(\frac{N}{2}+1\right)} d\mu^g\right)^{2/N} \leq \int_M A^2 d\mu^g.
\]

As a consequence, we have obtained the following result:
\begin{lemma}\label{lmLichnerowicz}
    There exists a positive constant $\displaystyle \mu_L \definedas \nu \left(\frac{\min_M \psi}{\max_M \psi}\right)^N$ such that, if $\phi$ denotes the solution to the Lichnerowicz equation~\eqref{eqLichGeneral}, we have
    \[
        \mu_L \|\phi^N\|_{L^{\frac{N}{2}+1}}^{2\frac{n-1}{n}} \leq \|A\|_{L^2}^2.
    \]
\end{lemma}

When $A$ enjoys more regularity, elliptic regularity can be applied leading to
the following result:

\begin{proposition}[Lemma 2.11 and Proposition 2.12 of~\cite{Pailleron}]\label{propRegLich}
    Assume that $A \in L^{2r}(M, \bR)$ for $r \in (1, p]$. Then the solution $\phi$ to the Lichnerowicz equation~\eqref{eqLichGeneral} is a strong solution that belongs to $W^{2, r}(M, \bR)$. Further $\phi^N \in L^t(M, \bR)$ with
    \[
        t =
        \left\lbrace
        \begin{aligned}
            \frac{2(n-1)r}{n-2r}             & \quad \text{if } r < \frac{n}{2}, \\
            \text{arbitrary in } (1, \infty) & \quad \text{if } r > \frac{n}{2}.
        \end{aligned}
        \right.
    \]
\end{proposition}

Next, concerning the vector equation~\eqref{eqVector}, the case of manifolds
with boundary presents significantly more complexity compared to the
boundaryless case. Compare for example~\cite[Proposition 3.1]{Pailleron}
and~\cite[Proposition 5]{MaxwellNonCMC}. The reason is to be found in the fact
that the boundary conditions that are natural for the vector equation do not
interact well with Bochner's formula for the conformal Killing operator. The
result we will need is the following:

\begin{proposition}\label{propVector}
    Assume that $(M, g)$ has no non-zero conformal Killing vector field. Then the operator $\DeltaL: W^{2, q}(M, TM) \to L^q(M, TM)$ is invertible for any $q \in (1, p]$.
\end{proposition}

Note that, however, the existence of the inverse of $\DeltaL$ is obtained by
non-constructive methods such as the open mapping theorem. Hence, its
(operator) norm can be hard to control. We give a simpler estimate, with a
constant that can be more easily estimated, that will suffice our purpose apart
from the final issue of regularity of the solution where explicit constants are
less important.

It is a classical fact that the non-existence of conformal Killing vector
fields is equivalent to the coercivity of the quadratic form
\[
    q(V) \definedas \frac{1}{2} \int_M |\bL V|^2 d\mu^g
\]
over the space $W^{1, 2}(M, TM)$, see e.g.~\cite[Appendix
    A]{DahlGicquaudHumbert}. Hence, there exists a constant $\kappa > 0$ such that
\[
    \forall V \in W^{1, 2}(M, TM),~ \kappa \|V\|_{W^{1, 2}} \leq q(V).
\]
Applying the Sobolev embedding theorem for vector fields: $W^{1, 2}(M, TM)
    \hookrightarrow L^N(M, TM)$, we conclude that there exists a constant $\mu_0 >
    0$ such that
\[
    \forall V \in W^{1, 2}(M, TM),~ \mu_0 \|V\|_{L^N}^2 \leq q(V).
\]
The best choice for this constant is given by the Yamabe-like invariant $\mu_V$
given in \eqref{eqYamabeVector}.

We now show how Lemma~\ref{lmLichnerowicz} and the definition of the constant
$\mu_V$ in~\eqref{eqYamabeVector} imply the existence of solutions to the
conformal constraint equations when $\sigma$ has a small $L^2$-norm. We choose
for the Banach space $X$, which we will apply the Schauder fixed point on, the
space
\[
    X = L^r(M, \bR), \text{ with $r$ such that } \frac{1}{r} + \frac{1}{q} + \frac{1}{N} = 1,
\]
where $q$ was defined at the beginning of this section. A short calculation
shows that, as $q > n$, we have
\[
    r < \frac{N}{2} + 1,
\]
matching the requirements for the map $\mathrm{Lich}$ to be continuous. This
choice for $r$ is motivated by the following claim:
\begin{claim}\label{cl1}
    The mapping $\mathrm{Vect}: L^r(M, \bR) \to L^2(M, \mathring{S}_2 M)$ sending a given function $u \in L^r(M, \bR)$ to $\bL W$, where $W$ is the solution of the following vector equation
    \begin{equation}\label{eqDummy}
        \DeltaL W = \frac{n-1}{n} u d\tau
    \end{equation}
    is linear and continuous with norm not greater than $\displaystyle \frac{n-1}{n} \sqrt{\frac{2}{\mu_V}} \|d\tau\|_{L^q}$:
    \[
        \|\bL W\|_{L^2} \leq \frac{n-1}{n}\sqrt{\frac{2}{\mu_V}} \|d\tau\|_{L^q} \|u\|_{L^r}.
    \]
\end{claim}

\begin{proof}
    Note that $\displaystyle u d\tau \in L^t(M, T^*M)$ with $\displaystyle \frac{1}{t} = \frac{1}{r} + \frac{1}{q} = 1 - \frac{1}{N}$ so Proposition~\ref{propVector} applies to provide a unique solution $W \in W^{2, t}(M, TM)$ to~\eqref{eqDummy}\footnote{We are not forced to rely on Proposition~\ref{propVector} at this point as we could also use the Lax-Milgram theorem.}. To estimate the $L^2$-norm of $\bL W$, we contract~\eqref{eqDummy} with $W$ and integrate over $M$:
    \[
        \int_M \<W, \DeltaL W\>~d\mu^g = \frac{n-1}{n} \int_M u \<d\tau, W\>~d\mu^g.
    \]
    Integrating by parts the left hand side, we obtain:
    \[
        - \frac{1}{2} \int_M |\bL W|^2 d\mu^g = \frac{n-1}{n} \int_M u \<d\tau, W\>~d\mu^g.
    \]
    Hence, it follows from H\"older's inequality that
    \begin{align}
        \frac{1}{2} \int_M |\bL W|^2 d\mu^g
         & = -  \frac{n-1}{n} \int_M u \<d\tau, W\>~d\mu^g\nonumber                             \\
         & \leq \frac{n-1}{n} \|u\|_{L^r} \|d\tau\|_{L^q} \|W\|_{L^N}\label{eqEstimateVector0}.
    \end{align}
    From the definition~\eqref{eqYamabeVector} of the constant $\mu_V$, we have
    \[
        \|W\|_{L^N} \leq \left(\frac{1}{2\mu_V} \int_M |\bL W|^2 d\mu^g\right)^{1/2}.
    \]
    As a consequence,~\eqref{eqEstimateVector0} implies
    \[
        \frac{1}{2} \int_M |\bL W|^2 d\mu^g
        \leq \frac{n-1}{n} \|u\|_{L^r} \|d\tau\|_{L^q} \left(\frac{1}{2\mu_V} \int_M |\bL W|^2 d\mu^g\right)^{1/2}.
    \]
    Squaring this inequality and dividing both sides by $\displaystyle \frac{1}{2}
        \int_M |\bL W|^2 d\mu^g$, we get
    \[
        \frac{\mu_V}{2} \int_M |\bL W|^2 d\mu^g
        \leq \left(\frac{n-1}{n}\right)^2 \|u\|_{L^r}^2 \|d\tau\|_{L^q}^2.
    \]
    This concludes the proof of the claim.
\end{proof}

We now define the mapping $F: L^r(M, \bR) \to L^r(M, \bR)$ as follows. Given $u
    \in L^r(M, \bR)$, we let $W = \mathrm{Vect}(u)$ be the solution
to~\eqref{eqDummy} and set $F(u) = \phi^N$, where $\phi$ is the solution
to~\eqref{eqLichGeneral} with $A \definedas |\sigma + \bL W|$, i.e. $F(u) =
    \mathrm{Lich}(|\sigma + \bL W|)$. From Claim~\ref{cl1} and
Propostion~\ref{propLich}, $F$ is continuous and compact. We now prove that, if
$\sigma$ has a small enough $L^2$-norm, a certain closed ball in $L^r(M, \bR)$
is stable by $F$:
\begin{claim}\label{cl2}
    There exists an explicit $x_{\max} > 0$ given below in Equation~\eqref{eqDefSmax} such that, if
    \[
        \int_M |\sigma|^2 d\mu^g \leq  x_{\max},
    \]
    with $\sigma \not\equiv 0$, the closed ball
    \[
        \Omega \definedas \left\{u \in L^r(M, \bR), \|u\|_{L^r} \leq R_{\mathrm{opt}}\right\}
    \]
    with $R_{\mathrm{opt}}$ defined in~\eqref{eqDefR0} is stable for $F$.
\end{claim}

\begin{proof}
    Assume given $u \in L^r(M, \bR)$ such that $\|u\|_{L^r} \leq R$, where, as indicated in the statement of the claim, $R$ will be chosen below. From Claim~\ref{cl1}, we have that the solution $W$ to~\eqref{eqDummy} satisfies
    \[
        \|\bL W\|_{L^2}^2 \leq \frac{2}{\mu_V} \left(\frac{n-1}{n}\right)^2 \|d\tau\|_{L^q}^2 R^2.
    \]
    Because of the $L^2$-orthogonality between TT-tensors and tensors of the form
    $\bL V$, we have
    \begin{align*}
        \|A\|_{L^2}^2
         & = \int_M |\sigma + \bL W|^2 d\mu^g                                                                  \\
         & = \int_M |\sigma|^2 d\mu^g + \int_M |\bL W|^2 d\mu^g                                                \\
         & \leq \int_M |\sigma|^2 d\mu^g + \frac{2}{\mu_V} \left(\frac{n-1}{n}\right)^2 \|d\tau\|_{L^q}^2 R^2.
    \end{align*}
    Set $\displaystyle x \definedas \int_M |\sigma|^2 d\mu^g = \|\sigma\|_{L^2}^2$ so the previous inequality reads
    \[
        \|A\|_{L^2}^2
        \leq x + \frac{2}{\mu_V} \left(\frac{n-1}{n}\right)^2 \|d\tau\|_{L^q}^2 R^2.
    \]
    Lemma~\ref{lmLichnerowicz} then implies that $F(u) = \phi^N$, with $\phi$ the
    solution to~\eqref{eqLichGeneral}, satisfies
    \[
        \mu_L \|F(u)\|_{L^{\frac{N}{2}+1}}^{2\frac{n-1}{n}} \leq \|A\|_{L^2}^2 \leq x + \frac{2}{\mu_V} \left(\frac{n-1}{n}\right)^2 \|d\tau\|_{L^q}^2 R^2.
    \]
    Note that we used here the assumption $\sigma \not\equiv 0$ to conclude that $A
        \not\equiv 0$ as assumed in Lemma~\ref{lmLichnerowicz}. From the fact that $(M,
        g)$ has volume $1$, we have
    \[
        \|F(u)\|_{L^r} \leq \|F(u)\|_{L^{\frac{N}{2}+1}}.
    \]
    Combining the previous two estimates, we obtain:
    \[
        \mu_L \|F(u)\|_{L^r}^{2\frac{n-1}{n}} \leq x + \frac{2}{\mu_V} \left(\frac{n-1}{n}\right)^2 \|d\tau\|_{L^q}^2 R^2.
    \]
    So, we are guaranteed to have $\|F(u)\|_{L^r} \leq R$, i.e. that $\Omega$ is
    stable for $F$, provided that
    \begin{equation}\label{eqStable}
        x + \frac{2}{\mu_V} \left(\frac{n-1}{n}\right)^2 \|d\tau\|_{L^q}^2 R^2 \leq
        \mu_L R^{2\frac{n-1}{n}}.
    \end{equation}
    Now remark that the function
    \begin{equation}\label{eqDefH}
        h(R) \definedas x + \frac{2}{\mu_V} \left(\frac{n-1}{n}\right)^2 \|d\tau\|_{L^q}^2 R^2 - \mu_L R^{2\frac{n-1}{n}}
    \end{equation}
    is minimal for $R = R_{\mathrm{opt}}$ with
    \begin{equation}\label{eqDefR0}
        R_{\mathrm{opt}} = \left(\frac{2}{\mu_L \mu_V} \frac{n-1}{n} \|d\tau\|_{L^q}^2\right)^{-\frac{n}{2}}
    \end{equation}
    as it follows from computing the value for which $h'(R) = 0$. This choice for $R$ is the best possible in the sense that if~\eqref{eqStable} is not satisfied for this particular value, it is not satisfied for any other. Hence, provided that $h(R) \leq 0$, we have that $\|F(u)\|_{L^r} \leq R$ meaning that the ball $\Omega$ is stable for $F$. The function $h$ defined in~\eqref{eqDefH} depends linearly on $x$ and is decreasing for small values of $R$. Hence, if $x$ is less than $x_{\max}$, where $x_{\max}$ is given By
    \begin{equation}\label{eqDefSmax}
        x_{\max} \definedas \mu_L R_{\mathrm{opt}}^{2\frac{n-1}{n}} - \frac{2}{\mu_V} \left(\frac{n-1}{n}\right)^2 \|d\tau\|_{L^q}^2 R_{\mathrm{opt}}^2,
    \end{equation}
    we have $h(R_{\mathrm{opt}}) \leq 0$.
\end{proof}

We now apply the Schauder-Tychonoff fixed point theorem to the function $F$ and
the closed subset $\Omega$:

\begin{theorem}\label{thmSchauder}
    Let $\Omega$ be a bounded closed convex subset of a Banach space $X$. Assume that $F: \Omega \to \Omega$ is a continuous and compact mapping. Then $F$ admits a fixed point on $\Omega$.
\end{theorem}

This form of the fixed point theorem follows from the more classical one where
$F$ is only assumed to be continuous but $\Omega$ is assumed to be convex and
compact (see e.g.~\cite[Theorem 11.1]{GilbargTrudinger}) by replacing the set
$\Omega$ by the closed convex hull $\mathrm{cl}(F(\Omega))$ of $F(\Omega)$
which is know to be compact by the Mazur's theorem on convex
hulls~\cite[Theorem 3.25]{Rudin}. As, by assumption $F(\Omega) \subset \Omega$
with $\Omega$ convex and closed, we have $\mathrm{cl}(F(\Omega)) \subset
    \Omega$ which shows that $\mathrm{cl}(F(\Omega))$ is stable by $F$.

From Claim~\ref{cl2}, if the assumption of Claim~\ref{cl2} is fulfilled,
i.e.~that~\eqref{eqXmax} holds, the ball $\Omega$ is stable for $F$. As a
consequence, from Theorem~\ref{thmSchauder}, there exists a fixed point $\phi^N
    \in L^r(M, \bR)$ for $F$. This means that, setting $W = \mathrm{Vect}(\phi^N)$,
the pair $(\phi, W)$ is a weak solution to the conformal constraint
equations~\eqref{eqLichnerowicz}-\eqref{eqVector}.

Thus, the only points that are left unproven are the fact that $\phi$ is a
strong solution to the Lichnerowicz equation~\eqref{eqLichnerowicz} and the
announced regularities for $\phi$ and $W$.

We already know that $\phi^N \in L^{\frac{N}{2}+1}(M, \bR)$ as it follows from
Proposition~\ref{propLich}. We show inductively that $\phi^N \in L^{r_k}(M,
    \bR)$ for larger and larger values of $r_k$, starting from $r_0 = \frac{N}{2} +
    1$. Assume proven that $\phi^N \in L^{r_k}(M, \bR)$ with $r_k \geq r_0$. We
have that
\[
    \DeltaL W = \frac{n-1}{n} \phi^N d\tau \in L^{s_k}(M, T^*M)
\]
with $s_k$ satisfying $\displaystyle \frac{1}{s_k} = \frac{1}{r_k} +
    \frac{1}{q} < 1$. Hence $\bL W \in W^{1, s_k}(M, \mathring{S}_2M)
    \hookrightarrow L^{s'_k}(M, \mathring{S}_2M)$, with $s'_k$ being given by
\[
    \frac{1}{s'_k} = \frac{1}{s_k} - \frac{1}{n} = \frac{1}{r_k} + \frac{1}{q} - \frac{1}{n}
\]
if $s_k < n$ and $s'_k = \infty$ if $s_k > n$. Hence, $A \in L^{t_k}(M, \bR)$
with $t_k = \min\{2p, s'_k\}$. From Proposition~\ref{propRegLich}, we conclude
that
\[
    \phi^N \in L^{r_{k+1}}(M, \bR)
\]
with $r_{k+1} = \frac{2(n-1) t_k}{n - t_k}$. Assuming that $s_k < n$ and $s'_k
    < 2p$, we get the following formula for $r_{k+1}$:
\[
    \frac{1}{r_{k+1}} = \frac{n}{2(n-1)} \left(\frac{1}{r_k} + \frac{1}{q} - \frac{1}{n}\right) - \frac{1}{2(n-1)}.
\]
So $\frac{1}{r_k}$ satisfies a linear recurrence relation converging to a
negative value. This shows that one of the assumptions $s_k < n$ or $s'_k < 2p$
must be violated for some value $k > 0$. In either case, we see that we have
$\bL W \in L^{2p}(M, \mathring{S}_2 M)$. Proposition~\ref{propRegLich} then
shows that $\phi \in W^{2, p}(M, \bR) \hookrightarrow L^\infty(M, \bR)$. In
particular, the right hand side of the vector equation~\eqref{eqVector} belongs
to $L^q(M, T^* M)$. Proposition~\ref{propVector} finally applies to show that
$W \in W^{2, p}(M, TM)$ (note that the regularity of $W$ is limited here by the
regularity of the metric $g$).

We would like to make a couple of important remarks concerning the statement of
Theorem~\ref{thmA}:
\begin{enumerate}[leftmargin=*]
    \item Firstly, the condition~\eqref{eqStable} is fulfilled for some $R > 0$ not only
          if $x = \int_M |\sigma|^2 d\mu^g$ is small but also if $\|d\tau\|_{L^q}$ is
          small. As a consequence, Theorem~\ref{thmA} can also be understood as a
          near-CMC existence result in the same vein as~\cite{ACI08}.
    \item Secondly, we took for $R_{\mathrm{opt}}$ in Claim~\ref{cl2} the optimal value,
          i.e.~the one that leads to the largest possible range for $x$. We might also be
          interested in the smallest value of $R$ such that the stability
          condition~\eqref{eqStable} is fulfilled. This gives an hint on the size of the
          solution(s) $(\phi, W)$ to the conformal constraint equations~\eqref{system}
          with the smallest norm, i.e.~such that $\|\phi^N\|_{L^r}$ is minimal. These are
          the solutions that are the most natural from an analytic perspective as
          $\sigma$ should be thought of as a sort of a source term in the
          system~\eqref{system} moving the solution away from $\phi \equiv
              0$\footnote{Note that, due to the negative exponent of $\phi$ in the
              Lichnerowicz equation, having $\phi \equiv 0$ is meaningless. However, this
              makes sense if we consider instead the Lichnerowicz equation multiplied by
              $\phi^{N+1}$.}. Due to the fact that the powers of $R$ appearing
          in~\eqref{eqStable} are $R^2$ and $R^{2\frac{n-1}{n}}$ and since
          $2\frac{n-1}{n} < 2$, we see that for small values of $R$ the term proportional
          to $R^2$ is negligible compared to the other two. As a consequence, we have
          that the smallest value $R_{\min}$ of $R$ such that~\eqref{eqStable} holds is
          approximately given by
          \begin{equation}\label{eqEstimatedNorm}
              \mu_L R_{\min}^{2\frac{n-1}{n}} \simeq x.
          \end{equation}
          The first aim of the next section is to show that solutions to the conformal constraint equations with volume bounded by a certain constant indeed satisty $\|\phi^N\|_{L^r}^{2 \frac{n-1}{n}} \lesssim x$ (see Lemma~\ref{lmEstimatePhi}).
\end{enumerate}

%% file: uniqueness.tex
We now come to the proof of Theorem~\ref{thmB}. First note that the existence
part of the theorem follows from the previous one. The proof of uniqueness goes
as usual by considering two solutions $(\phi_1, W_1)$ and $(\phi_2, W_2)$ to
the conformal constraint equations with volume bounded by $V_{\max}$ and
proving that they are equal. The main novel ingredient is an estimate for the
difference between $\phi_1$ and $\phi_2$ that will appear in
Proposition~\ref{propMaxPrinciple}. Note that~\cite{ACI08} also contains a
uniqueness statement. However, their hypotheses are different as they assume an
upper bound of the form $|d\tau| \leq c \min_M \tau$ for some positive constant
$c$ while we are making no assumption on $\tau$ here but on the relative size
of $\sigma$, $\essinf_M |\sigma|$ and $d\tau$.

Our first objective is to find an estimate (in some Lebesgue norm) for
solutions with small volume. This part can be carried with low regularity
assumptions, notably for $\sigma$, as in the previous section, so we continue
using the same kind of methods.

\begin{lemma}
    Let $V_{\max} > 0$ be given and let $(\phi, W)$ be a solution to the conformal constraint equations such that
    \[
        \|\phi^N\|_{L^1} = \vol_{\ghat}(M) \leq V_{\max}.
    \]
    Then we have
    \[
        \int_M |\bL W|^2 d\mu^g
        \leq \frac{2}{\mu_V} \left(\frac{n-1}{n}\right)^2 V_{\max}^{2/n} \|\phi^N\|_{L^{\frac{N}{2} + 1}}^{2\frac{n-1}{n}} \|d\tau\|_{L^n}^2.
    \]
\end{lemma}

\begin{proof}
    As before, we multiply the vector equation~\eqref{eqVector} by $W$ and integrate over $M$ to get
    \begin{equation}\label{eqEstimateVector1}
        \begin{aligned}
            \frac{1}{2} \int_M |\bL W|^2 d\mu^g
             & = -\frac{n-1}{n} \int_M \phi^N \<d\tau, W\> d\mu^g               \\
             & \leq \frac{n-1}{n} \|\phi^N\|_{L^2} \|d\tau\|_{L^n} \|W\|_{L^N},
        \end{aligned}
    \end{equation}
    We next write
    \[
        \frac{1}{2} = \frac{1-\lambda}{1} + \frac{\lambda}{\frac{N}{2} + 1},
    \]
    with $\lambda = \frac{n-1}{n}$. Using the interpolation inequality, we get
    \[
        \|\phi^N\|_{L^2} \leq \|\phi^N\|_{L^1}^{1/n} \|\phi^N\|_{L^{\frac{N}{2} + 1}}^{\frac{n-1}{n}} \leq V_{\max}^{1/n} \|\phi^N\|_{L^{\frac{N}{2} + 1}}^{\frac{n-1}{n}}.
    \]
    Hence, the estimate~\eqref{eqEstimateVector1} implies
    \[
        \frac{1}{2} \int_M |\bL W|^2 d\mu^g
        \leq \frac{n-1}{n} V_{\max}^{1/n} \|\phi^N\|_{L^{\frac{N}{2} + 1}}^{\frac{n-1}{n}} \|d\tau\|_{L^n} \|W\|_{L^N}.
    \]
    Proceeding as in the proof of Claim~\ref{cl1}, we obtain the claimed estimate:
    \[
        \int_M |\bL W|^2 d\mu^g
        \leq \frac{2}{\mu_V} \left(\frac{n-1}{n}\right)^2 V_{\max}^{2/n} \|\phi^N\|_{L^{\frac{N}{2} + 1}}^{2\frac{n-1}{n}} \|d\tau\|_{L^n}^2.
    \]
\end{proof}

We promote the previous result to an estimate for
$\|\phi^N\|_{L^{\frac{N}{2}+1}}$. At first, note that, setting $A^2 = |\sigma +
    \bL W|^2$, we have
\begin{align*}
    \|A\|_{L^2}^2
     & = \int_M |\sigma + \bL W|^2 d\mu^g                                                                                                        \\
     & = \int_M |\sigma|^2 d\mu^g + \int_M |\bL W|^2 d\mu^g                                                                                      \\
     & \leq x + \frac{2}{\mu_V} \left(\frac{n-1}{n}\right)^2 V_{\max}^{2/n} \|\phi^N\|_{L^{\frac{N}{2} + 1}}^{2\frac{n-1}{n}} \|d\tau\|_{L^n}^2,
\end{align*}
where we used the fact that the TT-tensor $\sigma$ is $L^2$-orthogonal to $\bL W$ and were we set $x \definedas \|\sigma\|_{L^2}^2$. We can now obtain the desired estimate for $\phi^N$. From Lemma~\ref{lmLichnerowicz} together with the previous result, we get
\[
    \mu_L \|\phi^N\|_{L^{\frac{N}{2}+1}}^{2 \frac{n-1}{n}}
    \leq \|A\|_{L^2}^2 \leq x + \frac{2}{\mu_V} \left(\frac{n-1}{n}\right)^2 V_{\max}^{2/n} \|\phi^N\|_{L^{\frac{N}{2} + 1}}^{2\frac{n-1}{n}} \|d\tau\|_{L^n}^2.
\]
We have proven the following result:

\begin{lemma}\label{lmEstimatePhi}
    Assume that $(\phi, W)$ is a solution to the conformal constraint equations such that the physical volume $\vol_{\ghat}(M)$ is bounded by $V_{\max}$ with
    \begin{equation}\label{eqUpperboundVmax}
        \frac{2}{\mu_V} \left(\frac{n-1}{n}\right)^2 \|d\tau\|_{L^n}^2 V_{\max}^{2/n} < \mu_L.
    \end{equation}
    Then we have
    \begin{equation}\label{eqEstimatePhi2}
        \left(\mu_L - \frac{2}{\mu_V} \left(\frac{n-1}{n}\right)^2 \|d\tau\|_{L^n}^2 V_{\max}^{2/n}\right) \|\phi^N\|_{L^{\frac{N}{2}+1}}^{2 \frac{n-1}{n}} \leq \|\sigma\|_{L^2}^2.
    \end{equation}
\end{lemma}
This estimate should not be overlooked. It says that, as soon as we impose an upper bound on the volume of the solutions $(\phi, W)$ to the conformal constraint equations so that~\eqref{eqUpperboundVmax} is fulfilled, the solutions automatically satisfies an estimate similar to~\eqref{eqEstimatedNorm} (with a lesser constant). In particular, as $x$ becomes small, all solutions to the conformal constraint equation with volume bounded by $V_{\max}$ are actually ``small solutions'' with their size predicted (up to some multiplicative constant) by~\eqref{eqEstimatedNorm}.

We promote the estimate~\eqref{eqEstimatePhi2} to a $L^\infty$-bound for $\bL
    W$. The reason for this will become appearant in the proof of
Theorem~\ref{thmB}. Note that we will now be using the $L^{2p}$-norm of
$\sigma$ instead of its $L^2$-norm. This is because of technical difficulties
that we discuss at the end of the paper.

\begin{proposition}\label{propEstimatePhi}
    Assume that $\delta$ is strictly positive, where $\delta$ is defined in~\eqref{eqDefDelta}. Then, there exists a constant $C > 0$ such that the following holds. For any TT-tensor $\sigma$ such that $\|\sigma\|_{L^{2p}} \leq 1$, if $(\phi, W)$ is a solution to the conformal constraint equations with volume bounded by $V_{\max}$, we have
    \[
        \left\|\bL W\right\|_{L^\infty} \leq C \|\sigma\|_{L^{2p}}^{\frac{n}{n-1}}.
    \]
    Further, under the same assumptions, for any $r \in (1, \infty)$, there exists
    a constant $C' = C'(r)$ independent of $\sigma$, $\phi$ and $W$ such that
    \[
        \left\|\phi^N\right\|_{L^r} \leq C' \|\sigma\|_{L^{2p}}^{\frac{n}{n-1}}.
    \]
\end{proposition}

\begin{proof}
    The proof is done by a bootstrap argument akin to the one in the proof of Theorem~\ref{thmA}. We will make the simplifying assumption that $\scal \geq \epsilon > 0$. The general case can be handled as in Proposition~\ref{propLich}. From Lemma~\ref{lmEstimatePhi}, we have a bound on $\|\phi^N\|_{L^{\frac{N}{2}+1}}$:
    \[
        \|\phi^N\|_{L^{\frac{N}{2}+1}}^{2 \frac{n-1}{n}} \leq \delta^{-1} \|\sigma\|_{L^2}^2 \leq \delta^{-1} \|\sigma\|_{L^{2p}}^2.
    \]
    We set $r_0 = \frac{N}{2} + 1$ and construct inductively an increasing sequence
    $(r_k)_k$ together with estimates for $\phi^N$ in $L^{r_k}(M, \bR)$:
    \begin{equation}\label{eqInduction}
        \|\phi^N\|_{L^{r_k}}^{\frac{n-1}{n}} \leq \mu_k \|\sigma\|_{L^{2p}},
    \end{equation}
    for some constant $\mu_k > 0$ independent of $\sigma$, $\phi$ and $W$. The previous estimate shows that it holds true for $k = 0$ with $\mu_0 = \delta^{-1/2}$. From Proposition~\ref{propVector}, we have
    \begin{equation}\label{eqInductionW}
        \|W\|_{W^{2, s_k}} \lesssim \|\phi^N\|_{L^{r_k}} \|d\tau\|_{L^q}
    \end{equation}
    where $s_k$ is such that
    \[
        \frac{1}{s_k} = \frac{1}{r_k} + \frac{1}{q}.
    \]
    Assuming that $s_k < n$, we deduce that, for some constant $C_k$,
    \[
        \|\bL W\|_{L^{t_k}} \leq C_k \|\phi^N\|_{L^{r_k}} \|d\tau\|_{L^q},
    \]
    where $t_k$ is defined as follows:
    \[
        \frac{1}{t_k} = \frac{1}{s_k} - \frac{1}{n} = \frac{1}{r_k} + \frac{1}{q} - \frac{1}{n}.
    \]
    Setting $A \definedas |\sigma + \bL W|$, we have
    \[
        \|A\|_{L^{t_k}}
        \leq \|\sigma\|_{L^{t_k}} + \|\bL W\|_{L^{t_k}}
        \lesssim \|\sigma\|_{L^{2p}} + \|\phi^N\|_{L^{r_k}} \|d\tau\|_{L^q}.
    \]
    We next choose $\psi = \phi^{N+1+2\ell}$ as a test function for the
    Lichnerowicz equation~\eqref{eqLichGeneral} for some $\ell \geq 0$ to be chosen
    later. Proceeding as in the proof of Proposition~\ref{propLich}, we get
    \[
        \int_M \left[- \frac{4(n-1)}{n-2} \frac{N+1+2\ell}{\left(\frac{N}{2} + 1 + \ell\right)^2} \left|d \phi^{\frac{N}{2}+1+\ell}\right|^2 + \scal~\phi^{N+2+2\ell}\right] d\mu^g \leq \int_M A^2 \phi^{2\ell} d\mu^g.
    \]
    From the Sobolev embedding, we obtain that
    \[
        \left\|\phi^{\frac{N}{2}+1+\ell}\right\|_{L^N}^2
        \lesssim \int_M A^2 \phi^{2\ell} d\mu^g
        \leq \|A\|_{L^{t_k}}^2 \|\phi^{2\ell}\|_{L^\beta},
    \]
    where $\beta$ satisfies
    \[
        1 = \frac{1}{\beta} + \frac{2}{t_k}.
    \]
    Rearranging to have $\phi^N$ appearing everywhere, we obtain
    \begin{equation}\label{eqDummy2}
        \left\|\phi^N\right\|_{L^{\frac{N}{2}+1+\ell}}^{\frac{N+2+2\ell}{N}}
        \lesssim \|A\|_{L^{t_k}}^2 \|\phi^N\|_{L^{\frac{2\ell\beta}{N}}}^{\frac{2\ell}{N}},
    \end{equation}
    We now choose $\ell$. The optimal choice would be the one such that $\frac{N}{2} + 1 + \ell = \frac{2 \beta}{N}\ell$ but this choice leads to a complicated formula for $r_{k+1}$ in terms of $r_k$ so it seems wiser to choose $\ell$ so that $\frac{2\beta}{N} \ell = r_k$. Some simple calculations lead to the following recurrence relation for $r_k$:
    \[
        r_{k+1} = \frac{N}{2} \left(1 + \frac{1}{n} - \frac{1}{q}\right) r_k + 1 - \frac{N}{2}.
    \]
    In particular, we remark that, as $q > n$, we have
    \[
        r_{k+1} - 1 \geq \frac{N}{2} \left(r_k - 1\right).
    \]
    Since $r_0 = \frac{N}{2} + 1 > 1$, we conclude that the sequence $(r_k)_k$
    grows exponentially fast. In particular, we have $r_{k+1} \geq r_k$.
    Consequently, in the estimate~\eqref{eqDummy2}, we have
    \[
        \left\|\phi^N\right\|_{L^{r_{k+1}}}^{\frac{N+2+2\ell}{N}}
        \lesssim \|A\|_{L^{t_k}}^2 \|\phi^N\|_{L^{r_{k+1}}}^{\frac{2\ell}{N}},
    \]
    which, after simplification, yields
    \[
        \left\|\phi^N\right\|_{L^{r_{k+1}}}^{2\frac{n-1}{n}} = \left\|\phi^N\right\|_{L^{r_{k+1}}}^{\frac{N+2}{N}} \lesssim \|A\|_{L^{t_k}}^2 \lesssim \|\sigma\|_{L^{2p}}^2 + \|d\tau\|_{L^q}^2 \|\phi^N\|_{L^{r_k}}^2.
    \]
    By induction, we have the estimate~\eqref{eqInduction}. Hence, we get
    \[
        \left\|\phi^N\right\|_{L^{r_{k+1}}}^{2\frac{n-1}{n}} = \left\|\phi^N\right\|_{L^{r_{k+1}}}^{\frac{N+2}{N}} \lesssim \|A\|_{L^{t_k}}^2 \leq 2\|\sigma\|_{L^{2p}}^2 + 2 C_k^2 \|d\tau\|_{L^q}^2 \left(\mu_k \|\sigma\|_{L^{2p}}\right)^{\frac{2n}{n-1}}.
    \]
    As we assumed that $\|\sigma\|_{L^{2p}} \leq 1$, we can bound the right term as
    follows:
    \[
        \left\|\phi^N\right\|_{L^{r_{k+1}}}^{2\frac{n-1}{n}} \lesssim \|\sigma\|_{L^{2p}}^2 + \mu_k^{\frac{2n}{n-1}}\|\sigma\|_{L^{2p}}^2.
    \]
    This completes the induction argument as we have shown that
    $\left\|\phi^N\right\|_{L^{r_{k+1}}}^{2\frac{n-1}{n}} \lesssim
        \|\sigma\|_{L^{2p}}^2$.

    As $(r_k)_k$ grows exponentially fast, there exists an index $k_0$ so that
    \[
        \frac{1}{s_k} = \frac{1}{r_k} + \frac{1}{q} < \frac{1}{n}.
    \]
    For this index $k_0$, the estimate~\eqref{eqInductionW} remains valid but, as
    $s_k > n$, we have $\bL W \in W^{1, s_k}(M, \Sring_2 M) \hookrightarrow
        L^\infty(M, \Sring_2 M)$. From the previous estimates, we conclude
    \[
        \|\bL W\|_{L^\infty}^{\frac{n-1}{n}} \lesssim \|\sigma\|_{L^{2p}}.
    \]
    We now return to the estimate~\eqref{eqDummy2} and note that we can replace
    $t_k$ by $2p$ and, hence, let $\beta$ be chosen so that $1 = \frac{1}{\beta} +
        \frac{1}{2p}$:
    \[
        \left\|\phi^N\right\|_{L^{\frac{N}{2}+1+\ell}}^{\frac{N+2+2\ell}{N}}
        \lesssim \|A\|_{L^{2p}}^2 \|\phi^N\|_{L^{\frac{2\ell\beta}{N}}}^{\frac{2\ell}{N}}
    \]
    Estimating as before $\|\phi^N\|_{L^{\frac{2\ell}{N}}}$ from above by
    $\left\|\phi^N\right\|_{L^{\frac{N}{2}+1+\ell}}$, we obtain
    \[
        \left\|\phi^N\right\|_{L^{\frac{N}{2}+1+\ell}}^{\frac{N+2}{N}}
        \lesssim \|A\|_{L^{2p}}^2.
    \]
    As we can choose $\ell$ as large as we want, this concludes the proof of the
    proposition.
\end{proof}

\begin{corollary}\label{corLowerBound}
    Let the assumptions of the previous proposition hold. Suppose in addition that
    \[
        \essinf_M |\sigma| \geq \theta \|\sigma\|_{L^{2p}}.
    \]
    Then, for \( \|\sigma\|_{L^{2p}} \) sufficiently small, any solution \( (\phi,
    W) \) to the conformal constraint equations with volume bounded by \( V_{\max}
    \) satisfies the lower bound
    \[
        A \definedas |\sigma + \bL W| \geq \frac{\theta}{2} \|\sigma\|_{L^{2p}}\quad \text{a.e.}.
    \]
\end{corollary}

\begin{proof}
    From Proposition~\ref{propEstimatePhi}, we have the pointwise estimate
    \[
        A \geq |\sigma| - |\bL W| \geq \theta \|\sigma\|_{L^{2p}} - C \|\sigma\|_{L^{2p}}^{\frac{n}{n-1}}.
    \]
    So the announced estimate follows as soon as we choose $\sigma$ so small that
    $\displaystyle C \|\sigma\|_{L^{2p}}^{\frac{1}{n-1}} \leq \frac{\theta}{2}$.
\end{proof}

We next prove an estimate for the difference between two solutions of the
Lichnerowicz equation:

\begin{proposition}\label{propMaxPrinciple}
    We have
    \begin{equation}\label{eqLichDiffInt0}
        \mu_L \left\|\phi_1 - \phi_2\right\|_{L^{N \left(\frac{N}{2}+1\right)}}^{N+2}
        \leq \int_M \left|A_1^{\frac{2}{N+2}} - A_2^{\frac{2}{N+2}}\right|^{N+2} d\mu^g,
    \end{equation}
    where $\mu_L > 0$ is the constant defined in Lemma~\ref{lmLichnerowicz}.
\end{proposition}

\begin{proof}
    We subtract the equations satisfied by $\phitil_1$ and $\phitil_2$, where $\phitil_i = \psi^{-1} \phi_i$ with $\psi$ as defined below Proposition~\ref{propLich}, namely
    \[
        - \frac{4(n-1)}{n-2} \Delta_{\gtil} \phitil_i + \scaltil~\phitil_i + \frac{n-1}{n} \tau^2 \phitil_i^{N-1} = \frac{\Atil_i^2}{\phitil_i^{N+1}},
    \]
    with $\Atil_i = \psi^{-N} |\sigma + \bL W_i|$. We obtain
    \[
        - \frac{4(n-1)}{n-2} \Delta_{\gtil} (\phitil_1 - \phitil_2) + \scaltil(\phitil_1 - \phitil_2) + \frac{n-1}{n} \tau^2 \left(\phitil_1^{N-1} - \phitil_2^{N-1}\right) = \frac{\Atil_1^2}{\phitil_1^{N+1}} - \frac{\Atil_2^2}{\phitil_2^{N+1}}.
    \]
    We multiply this equation by $(\phitil_1 - \phitil_2)^{N+1}_{0,}$, where we
    denote by $u_{0,}$ the positive part of $u$: $u_{0,} = \max\{u, 0\}$, in
    agreement with the notation from~\cite{Pailleron}. This operation is now
    allowed as we proved that $\phitil_1, \phitil_2 \in W^{2, p}(M, \bR)$.
    Integrating over $M$, we get
    \begin{equation}\label{eqLichDiffInt}
        \begin{aligned}
             & \int_M (\phitil_1 - \phitil_2)^{N+1}_{0,} \left[- \frac{4(n-1)}{n-2} \Delta_{\gtil} (\phitil_1 - \phitil_2) + \scaltil(\phitil_1 - \phitil_2)\right] d\mu^{\gtil} \\
             & \qquad = \int_M (\phitil_1 - \phitil_2)^{N+1}_{0,} \left[\frac{\Atil_1^2}{\phitil_1^{N+1}} - \frac{\Atil_2^2}{\phitil_2^{N+1}}\right] d\mu^{\gtil}                \\
             & \qquad\qquad  - \int_M (\phitil_1 - \phitil_2)^{N+1}_{0,} \left[\frac{n-1}{n} \tau^2 \left(\phitil_1^{N-1} - \phitil_2^{N-1}\right)\right] d\mu^{\gtil}.
        \end{aligned}
    \end{equation}
    Note that the second term in the right hand side is (pointwise) non-negative:
    \[
        \int_M (\phitil_1 - \phitil_2)^{N+1}_{0,} \left[\frac{n-1}{n} \tau^2 \left(\phitil_1^{N-1} - \phitil_2^{N-1}\right)\right] d\mu^{\gtil} \geq 0,
    \]
    since if $\phitil_1 - \phitil_2 \leq 0$, the integrand vanishes while if
    $\phitil_1 - \phitil_2 > 0$, both $(\phitil_1 - \phitil_2)^{N+1}_{0,}$ and
    $\phitil_1^{N-1} - \phitil_2^{N-1}$ are positive. Hence,
    Equation~\eqref{eqLichDiffInt} implies
    \[
        \begin{aligned}
             & \int_M \left[\frac{4(n-1)}{n-2} \left\<d (\phitil_1 - \phitil_2)^{N+1}_{0,}, d(\phitil_1 - \phitil_2)\right\>_{\gtil} + \scaltil(\phitil_1 - \phitil_2)^{N+2}_{0,}\right] d\mu^{\gtil} \\
             & \qquad \leq \int_M (\phitil_1 - \phitil_2)^{N+1}_{0,} \left[\frac{\Atil_1^2}{\phitil_1^{N+1}} - \frac{\Atil_2^2}{\phitil_2^{N+1}}\right] d\mu^{\gtil}.
        \end{aligned}
    \]
    The left hand side can be reorganized to yield
    \begin{equation}\label{eqLichDiffInt2}
        \begin{aligned}
             & \int_M \left[\frac{3n-2}{n-1} \left\<d (\phitil_1 - \phitil_2)^{\frac{N}{2}+1}_{0,}, d(\phitil_1 - \phitil_2)^{\frac{N}{2}+1}_{0,}\right\>_{\gtil} + \scaltil(\phitil_1 - \phitil_2)^{N+2}_{0,}\right] d\mu^{\gtil} \\
             & \qquad \leq \int_M (\phitil_1 - \phitil_2)^{N+1}_{0,} \left[\frac{\Atil_1^2}{\phitil_1^{N+1}} - \frac{\Atil_2^2}{\phitil_2^{N+1}}\right] d\mu^{\gtil}.
        \end{aligned}
    \end{equation}
    The idea is now to get an upper bound for the right hand side of~\eqref{eqLichDiffInt2}. To do so, we let $t \definedas \frac{\phitil_2}{\phitil_1}$ and note that, where ever $\phitil_1 \geq \phitil_2$, we have
    \[
        (\phitil_1 - \phitil_2)^{N+1}_{0,} \left[\frac{\Atil_1^2}{\phitil_1^{N+1}} - \frac{\Atil_2^2}{\phitil_2^{N+1}}\right] = (1 - t)^{N+1} \left[\Atil_1^2 - \frac{\Atil_2^2}{t^{N+1}}\right].
    \]
    As $\phitil_2 > 0$, $t$ ranges in $(0, 1]$. The function
    \[
        f(t) \definedas (1 - t)^{N+1} \left[\Atil_1^2 - \frac{\Atil_2^2}{t^{N+1}}\right]
    \]
    is increasing on the interval $(0, t_{\max}]$ with $t_{\max} =
        \left(\frac{\Atil_2}{\Atil_1}\right)^{\frac{1}{N+2}}$ and decreasing on the
    interval $[t_{\max}, \infty)$. Hence, we have
    \[
        f(t) \leq \left\lbrace
        \begin{aligned}
            f(t_{\max}) & = \left|\Atil_1^{\frac{2}{N+2}} - \Atil_2^{\frac{2}{N+2}}\right|^{N+2} & \qquad \text{if $\Atil_2 \leq \Atil_1$}, \\
            f(1)        & = 0                                                                    & \qquad \text{if $\Atil_2 > \Atil_1$}.
        \end{aligned}
        \right.
    \]
    As a consequence, the estimate~\eqref{eqLichDiffInt2} implies
    \[
        \begin{aligned}
             & \int_M \left[\frac{3n-2}{n-1} \left\<d (\phitil_1 - \phitil_2)^{\frac{N}{2}+1}_{0,}, d(\phitil_1 - \phitil_2)^{\frac{N}{2}+1}_{0,}\right\>_{\gtil} + \scaltil(\phitil_1 - \phitil_2)^{N+2}_{0,}\right] d\mu^{\gtil} \\
             & \qquad \leq \int_M \mathds{1}_{\{\Atil_1 \geq \Atil_2\}}\left[\Atil_1^{\frac{2}{N+2}} - \Atil_2^{\frac{2}{N+2}}\right] d\mu^{\gtil}.
        \end{aligned}
    \]
    Permutting the indices $1$ and $2$ and adding the corresponding estimate to the
    previous one, we obtain
    \[
        \begin{aligned}
             & \int_M \left[\frac{3n-2}{n-1} \left\<d (\phitil_1 - \phitil_2)^{\frac{N}{2}+1}, d(\phitil_1 - \phitil_2)^{\frac{N}{2}+1}\right\>_{\gtil} + \scaltil(\phitil_1 - \phitil_2)^{N+2}\right] d\mu^{\gtil} \\
             & \qquad \leq \int_M \left|\Atil_1^{\frac{2}{N+2}} - \Atil_2^{\frac{2}{N+2}}\right| d\mu^{\gtil}.
        \end{aligned}
    \]
    From the definition of the Sobolev constant $\nu$ in~\eqref{eqDefNu}, we
    conclude that
    \[
        \nu \left(\int_M (\phitil_1 - \phitil_2)^{N\left(\frac{N}{2}+1\right)} d\mu^{\gtil}\right)^{2/N} \leq \int_M \left|\Atil_1^{\frac{2}{N+2}} - \Atil_2^{\frac{2}{N+2}}\right| d\mu^{\gtil}.
    \]
    We now get back to the reference metric $g$:
    \[
        \nu \left(\int_M \psi^{-N^2/2} (\phi_1 - \phi_2)^{N\left(\frac{N}{2}+1\right)} d\mu^{\gtil}\right)^{2/N}
        \leq \int_M \psi^{-N} \left|A_1^{\frac{2}{N+2}} - A_2^{\frac{2}{N+2}}\right| d\mu^g.
    \]
    Thus,
    \[
        \nu \left(\frac{\max_M \psi}{\min_M \psi}\right)^N \left(\int_M (\phi_1 - \phi_2)^{N\left(\frac{N}{2}+1\right)} d\mu^{\gtil}\right)^{2/N}
        \leq \int_M \left|A_1^{\frac{2}{N+2}} - A_2^{\frac{2}{N+2}}\right| d\mu^g.
    \]
    The constant on the left is nothing but the constant $\mu_L$ defined in
    Lemma~\ref{lmLichnerowicz}. This concludes the proof of the proposition.
\end{proof}

We now have all the ingredients to prove Theorem~\ref{thmB}. We do this in a
series of claims. In what follows, we assume given two solutions $(\phi_1,
    W_1)$, $(\phi_2, W_2)$ to the conformal constraint equations with volume
bounded by $V_{\max}$, with $V_{\max}$ such that $\delta > 0$, where $\delta$
is as in~\eqref{eqDefDelta}. Set $A_i = |\sigma + \bL W_i|$ ($i=1, 2$). We also
assume that $\|\sigma\|_{L^{2p}}$ is small enough so that the assumptions of
Corollary~\ref{corLowerBound} are fulfilled.\\

From the previous proposition, we recall the estimate~\eqref{eqLichDiffInt0}:
\begin{equation}\label{eqLichDiffInt1}
    \left\|\phi_1 - \phi_2\right\|_{L^{N \left(\frac{N}{2}+1\right)}}
    \lesssim \left\|A_1^{\frac{2}{N+2}} - A_2^{\frac{2}{N+2}}\right\|_{L^{N+2}}.
\end{equation}

Our first task is to rework the right hand side of~\eqref{eqLichDiffInt0}:

\setcounter{claim}{0}
\renewcommand{\theclaim}{\arabic{claim}'}
\begin{claim}\label{cl1p}
    There exists an explicit constant $\kappa = \kappa(\theta)$ such that the following estimate holds independently of $\sigma$, $W_1$ and $W_2$:
    \[
        \left\|A_1^{\frac{2}{N+2}} - A_2^{\frac{2}{N+2}}\right\|_{L^{N+2}} \leq \kappa \|\sigma\|_{L^{2p}}^{-\frac{n}{2(n-1)}} \|\bL W_1 - \bL W_2\|_{L^{N+2}}.
    \]
\end{claim}

\begin{proof}
    We start by using the mean value theorem for the function $y \mapsto y^{\frac{2}{N+2}}$ on the interval $[A_2, A_1]$. We obtain
    \[
        \left|A_1^{\frac{2}{N+2}} - A_2^{\frac{2}{N+2}}\right| \leq \frac{2}{N+2} \left(\min\{A_1, A_2\}\right)^{-\frac{N}{N+2}} |A_1 - A_2|.
    \]
    From Corollary~\ref{corLowerBound}, we can estimate the minimum of $A_1$ and
    $A_2$ from below:
    \[
        \left|A_1^{\frac{2}{N+2}} - A_2^{\frac{2}{N+2}}\right| \leq \frac{2}{N+2} \left(\frac{\theta}{2} \|\sigma\|_{L^{2p}}\right)^{-\frac{N}{N+2}} |A_1 - A_2|.
    \]
    Hence, setting $\kappa \definedas \frac{2}{N+2}
        \left(\frac{\theta}{2}\right)^{-\frac{N}{N+2}}$, we have obtained
    \[
        \left|A_1^{\frac{2}{N+2}} - A_2^{\frac{2}{N+2}}\right|^{N+2} \leq \kappa^{N+2} \|\sigma\|_{L^{2p}}^{-N} |A_1 - A_2|^{N+2}.
    \]
    As a consequence, we get
    \[
        \left\|A_1^{\frac{2}{N+2}} - A_2^{\frac{2}{N+2}}\right\|_{L^{N+2}} \leq \kappa \|\sigma\|_{L^{2p}}^{-\frac{n}{2(n-1)}} \|A_1 - A_2\|_{L^{N+2}}.
    \]
    Finally, from the definition of $A_i$ and the triangle inequality, we remark
    that
    \[
        |A_1 - A_2| = \left||\sigma + \bL W_1| - |\sigma + \bL W_2|\right| \leq \left|(\sigma + \bL W_1) - (\sigma + \bL W_2)\right| = \left|\bL W_1 - \bL W_2\right|.
    \]
    This concludes the proof of the claim.
\end{proof}

\begin{claim}\label{cl2p}
    We have
    \[
        \left\|\bL W_1 - \bL W_2\right\|_{L^{N+2}} \lesssim \|\phi_1^N - \phi_2^N\|_{L^{N+2}} \|d\tau\|_{L^q}.
    \]
\end{claim}

\begin{proof}
    Note that $W_1 - W_2$ solves the following equation:
    \[
        \DeltaL (W_1 - W_2) = \frac{n-1}{n} (\phi_1^N - \phi_2^N) d\tau.
    \]
    The right hand side can be estimated using H\"older's inequality:
    \[
        \left\|\frac{n-1}{n} (\phi_1^N - \phi_2^N) d\tau\right\|_{L^s}
        \leq \frac{n-1}{n} \left\|\phi_1^N - \phi_2^N\right\|_{L^{N+2}} \|d\tau\|_{L^q},
    \]
    where $s \in (1, \infty)$ is such that
    \[
        \frac{1}{s} = \frac{1}{N+2} + \frac{1}{q}
    \]
    Hence, from Proposition~\ref{propVector}, we have $\|W_1 - W_2\|_{W^{2, s}}
        \lesssim \left\|\phi_1^N - \phi_2^N\right\|_{L^{N+2}} \|d\tau\|_{L^q}$. The
    conclusion of the claim now follows from the Sobolev embedding theorem: We have
    \[
        \left\|\bL W_1 - \bL W_2\right\|_{L^{N+2}} \leq \left\|\bL W_1 - \bL W_2\right\|_{L^{s'}} \lesssim \|\phi_1^N - \phi_2^N\|_{L^{N+2}} \|d\tau\|_{L^q},
    \]
    where $s' > N + 2$ is such that
    \[
        \frac{1}{s'} = \frac{1}{N+2} + \frac{1}{q} - \frac{1}{n}.
    \]
\end{proof}

Then we work on the left hand side of~\eqref{eqLichDiffInt1}:
\begin{claim}\label{cl3p}
    We have
    \[
        \|\phi_1^N - \phi_2^N \|_{L^{N+2}}
        \lesssim \|\phi_1 - \phi_2\|_{L^{N \left(\frac{N}{2}+1\right)}} \|\sigma\|_{L^{2p}}^{\frac{n+2}{2(n-1)}}.
    \]
\end{claim}

\begin{proof}
    From the mean value theorem, we have
    \begin{align*}
        \left|\phi_1^N - \phi_2^N \right|
         & \leq (N-1) \max\{\phi_1^{N-1}, \phi_2^{N-1}\} |\phi_1 - \phi_2| \\
         & \leq (N-1) (\phi_1^{N-1} + \phi_2^{N-1}) |\phi_1 - \phi_2|.
    \end{align*}
    From H\"older's inequality, we get
    \begin{align*}
        \left\|\phi_1^N - \phi_2^N \right\|_{L^s}
         & \leq (N-1) \|\phi_1 - \phi_2\|_{L^{N \left(\frac{N}{2}+1\right)}} \|\phi_1^{N-1} + \phi_2^{N-1}\|_{L^\alpha}                              \\
         & \leq (N-1) \|\phi_1 - \phi_2\|_{L^{N \left(\frac{N}{2}+1\right)}} \left(\|\phi_1^{N-1}\|_{L^\alpha} + \|\phi_2^{N-1}\|_{L^\alpha}\right),
    \end{align*}
    with $\alpha$ such that
    \[
        \frac{1}{r} = \frac{1}{N \left(\frac{N}{2}+1\right)} + \frac{1}{\alpha},
    \]
    i.e. $\alpha = \frac{N}{N-1} \left(\frac{N}{2} + 1\right)$. Changing the power
    of $\phi^{N-1}$ in the norms to $\phi^N$, we get
    \[
        \|\phi_1^N - \phi_2^N \|_{L^{\frac{N}{2}+1}}
        \leq (N-1) \|\phi_1 - \phi_2\|_{L^{N \left(\frac{N}{2}+1\right)}} \left(\|\phi_1^N\|_{L^{\frac{N}{2}+1}}^{\frac{N-1}{N}} + \|\phi_2^N\|_{L^{\frac{N}{2}+1}}^{\frac{N-1}{N}}\right).
    \]
    Next, we use the estimates obtained in Proposition~\ref{propEstimatePhi},
    namely $\displaystyle \|\phi_i^N\|_{L^r}^{2 \frac{n-1}{n}} \lesssim
        \|\sigma\|_{L^{2p}}^2$, and get
    \[
        \|\phi_1^N - \phi_2^N \|_{L^{\frac{N}{2}+1}}
        \lesssim \|\phi_1 - \phi_2\|_{L^{N \left(\frac{N}{2}+1\right)}} \|\sigma\|_{L^{2p}}^{\frac{n+2}{2(n-1)}}.
    \]
\end{proof}

We show how the three previous claims imply the theorem. From Claims~\ref{cl1p}
and~\ref{cl2p}, we get
\begin{align*}
    \left\|A_1^{\frac{2}{N+2}} - A_2^{\frac{2}{N+2}}\right\|_{L^{N+2}}
     & \leq \kappa \|\sigma\|_{L^{2p}}^{-\frac{n}{2(n-1)}} \|\bL W_1 - \bL W_2\|_{L^{N+2}}                 \\
     & \lesssim \|\sigma\|_{L^{2p}}^{-\frac{n}{2(n-1)}} \|\phi_1^N - \phi_2^N\|_{L^{N+2}} \|d\tau\|_{L^q}.
\end{align*}
From Proposition~\ref{propMaxPrinciple}, we obtain
\[
    \|\phi_1 - \phi_2\|_{L^{N \left(\frac{N}{2}+1\right)}}
    \lesssim \|\sigma\|_{L^{2p}}^{-\frac{n}{2(n-1)}} \|\phi_1^N - \phi_2^N\|_{L^{N+2}} \|d\tau\|_{L^q}.
\]
And, finally, from Claim~\ref{cl3p}, we conclude that
\[
    \|\phi_1^N - \phi_2^N\|_{L^{N + 2}}
    \leq K \|\sigma\|_{L^{2p}}^{\frac{1}{n-1}} \|\phi_1^N - \phi_2^N\|_{L^{N+2}} \|d\tau\|_{L^q},
\]
for some constant $K > 0$. As a consequence, if $K
    \|\sigma\|_{L^{2p}}^{\frac{1}{n-1}} \|d\tau\|_{L^q} < 1$, we have $\phi_1
    \equiv \phi_2$ and $W_1 \equiv W_2$ by Claim~\ref{cl2p}. This concludes the
proof of Theorem~\ref{thmB}.

\begin{remark}
    We conclude this paper with a couple of remarks.
    \begin{enumerate}[leftmargin=*]
        \item As we mentioned before stating Proposition~\ref{propEstimatePhi}, the beginning
              of the section proves estimates based on the $L^2$-norm of $\sigma$, while
              Proposition~\ref{propEstimatePhi} and Theorem~\ref{thmB} require a control on
              the $L^{2p}$-norm of $\sigma$. The reason for this is that the core of the
              proof of Theorem~\ref{thmB} consists in showing that, if $\sigma$ is small
              enough,
              \[
                  \|\phi_1^N - \phi_2^N\| \leq \epsilon \|A_1 - A_2\|,
              \]
              where $\epsilon = \epsilon(\sigma)$ tends to zero when $\|\sigma\|$ does (we do
              not make the norms precise here as we want to outline the argument. There could
              exist choices for the norms that are more relevant than the ones we made).
              However, the estimate~\eqref{eqLichDiffInt0} involves on its right hand side
              the difference $\displaystyle A_1^{\frac{2}{N+2}} - A_2^{\frac{2}{N+2}}$ which
              is the variation of the function $\displaystyle y \mapsto y^{\frac{2}{N+2}}$.
              The derivative of this function is $\frac{2}{N+2} y^{-\frac{N}{N+2}}$ which
              blows up at $y=0$. The condition that $\min_M |\sigma| \geq \theta
                  \|\sigma\|_{L^{2p}}$ ensures that both $A_1$ and $A_2$ are not too small as
              compared to their ``average value'' providing a uniform control of the right
              hand side.

              This difficulty could be overcame if we were able to estimate directly
              $\|\phi_1^N - \phi_2^N\|$. However, we were unable to achieve this goal.

        \item Despite multiple ways to construct TT-tensors, see e.g.~\cite{DelayTTTensors}
              and references therein, their zero set remains hard to control. In particular,
              the condition in Theorem~\ref{thmB} might only be obtained by a trial and error
              method. However, on Riemannian manifolds $(M, g)$ with symmetries (e.g. Lie
              groups with a left-invariant metric), there often exist TT-tensors with
              constant norm. Getting rid of the conformal Killing vector fields inherent to
              the symmetries might be obtained by taking suitable quotients. This is the
              point of view taken in~\cite{GicquaudConformal}.
    \end{enumerate}

\end{remark}